\documentclass[preprint]{elsarticle}
\usepackage{amssymb,amsthm,amsmath,enumerate,graphicx,mathtools,hyperref,amsfonts,latexsym,url,xcolor}
\usepackage[utf8]{inputenc}
\usepackage[T1]{fontenc} 
\usepackage[normalem]{ulem}
\usepackage[linesnumbered, ruled]{algorithm2e}

\usepackage{graphicx}
\usepackage{marvosym}	
\usepackage{xcolor}
\usepackage{csquotes}
\usepackage{float}
\usepackage{url}

\usepackage{tikz}
\usetikzlibrary{snakes,shapes}
\tikzstyle{tre}=[circle,draw,minimum size=3mm,inner sep=1pt]
\tikzstyle{trepp}=[circle,draw,minimum size=1.5mm,inner sep=0.1pt]
\tikzstyle{treppp}=[circle,draw,minimum size=1mm,inner sep=0pt]

\renewcommand{\leq}{\leqslant}
\renewcommand{\geq}{\geqslant}
\newcommand{\BT}{\mathcal{T}\!}
\newcommand{\TT}{\mathcal{T}\!}
\newcommand{\NN}{\mathbb{N}}
\newcommand{\RR}{\mathbb{R}}

\newcommand{\AAA}{\alpha}
\newcommand{\bal}{\mathrm{bal}}

\newcommand{\QC}{\mathit{QC}}

\theoremstyle{plain}
\newtheorem{theorem}{Theorem}
\newtheorem{lemma}[theorem]{Lemma}

\newtheorem{claim}{Claim}
\theoremstyle{definition}

\begin{document}
\begin{frontmatter}
\title{Squaring within the Colless index yields a better balance index}
\author[uib1]{Tom\'as M. Coronado}
\ead{t.martinez@uib.eu}
\author[uib1]{Arnau Mir}
\ead{arnau.mir@uib.eu}
\author[uib1]{Francesc Rossell\'o}
\ead{cesc.rossello@uib.es}
\address[uib1]{Dept. of Mathematics and Computer Science, University of the Balearic Islands, E-07122 Palma, Spain, and Balearic Islands Health Research Institute (IdISBa), E-07010 Palma, Spain}

\begin{abstract}
The Colless index for bifurcating phylogenetic trees, introduced by Colless \cite{Colless82}, is defined as the sum, over all internal nodes $v$ of the tree, of the absolute value of the difference of the sizes of the clades defined by the children of $v$. It  is one of the most popular phylogenetic balance indices, because, in addition to measuring the balance of a tree in a very simple and intuitive way, it turns out to be one of the most powerful and discriminating phylogenetic shape indices. But it has some drawbacks. On the one hand, although its minimum value is reached at the so-called maximally balanced trees, it is almost always  reached also at trees that are not maximally balanced. On the other hand, its definition as a sum of absolute values of differences makes it difficult to study analytically its  distribution  under probabilistic models of bifurcating phylogenetic trees. In this paper we show that if we replace in its definition the absolute values of the differences of clade sizes by the squares of these differences, all these drawbacks are overcome and the resulting index is still more powerful and discriminating than the original Colless index.
\end{abstract}
\end{frontmatter}

\section{Introduction}
Evolutionary biology is concerned, among other major things, about understanding what forces influence speciation and extinction processes, and how they affect macroevolution \citep{Futuyma}. In order to do so, there has been a natural interest in the development of techniques and measures whose goal is to assess the imprint of these forces in what has become the standard representation of joint evolutionary histories of groups of species: phylogenetic trees \citep{Kubo95,Mooers1997,Stich09}. There are two aspects of a phylogenetic tree that can expose such an imprint: its branch lengths ---determined by the timing of speciation events---, and its \emph{shape}, or \emph{topology} ---which, in turn, is determined by the differences in the diversification rates among clades \citep[Chap. 33]{fel:04}. But, as it turns out, the accurate reconstruction of branch lengths associating, to a given phylogenetic tree, a robust timeline is not straightforward \citep{Drummond} while, on the other hand, phylogenetic reconstruction methods over the same empirical data tend to agree on the topology of the reconstructed tree \citep{BrowerRindal13,Hillis92,RindalBrower11}. Therefore, the shape of phylogenetic trees has become the focus of most of the studies performed on this topic, be it \textit{via} the definition of indices quantifying topological features ---see,  for instance,  \citep{Fusco95,Mooers1997,Shao:90} and the references on balance indices given below--- or the frequency distribution of small rooted subtrees \citep{cherries,Savage,Slowinski90,Wu15}.

In his 1922 paper, Yule \citep{Yule} first observed that taxonomic trees have a tendency towards asymmetry, with most clades being small and only a few of them large at every taxonomic level. Thus, \emph{balance}, understood as the propensity of the children of any given node to have the same number of descendant leaves, has become the most popular topological measure used to describe the topology of a phylogenetic tree. Therefore, \textit{per negationem}, the \emph{imbalance} of a phylogenetic tree gives a measure of the tendency of diversification events to occur mostly along specific lineages \citep{Nelson,Shao:90}. Several such measures have been proposed, in order to quantify the balance (or, in many cases, the imbalance) of a phylogenetic tree, and they are referred to in the literature as \emph{balance indices}. For instance, see \citep{Colless82,CMR,Fischer2015,Fusco95,KiSl:93,cherries,Mir2013,Mir2018,Sackin:72,Shao:90} and the section ``Measures of overall asymmetry'' in \citep{fel:04} (pp.  562--563).

For instance, these indices have then been thoroughly used in order to test the validity evolutionary models \citep{Aldous01,Blum2005,duchene2018,KiSl:93,Mooers1997,Purvis1996,Verboom2019};  to assess possible biases in the distribution of shapes that are obtained through different phylogenetic tree reconstruction methods \citep{Colless1995,Farris98,Holton2014,Sober93,Stam02}; to compare tree shapes \citep{Avino18,Goloboff17,Kayondo}; as a tool to discriminate between input parameters in phylogenetic tree simulations 
\citep{Poon2015,Saulnier16}; or simply to describe phylogenies existing in the literature \citep{chalmandrier2018,Cunha2019,metzig2019,Purvis11}.

Introduced in \citep{Colless82}, the \emph{Colless index} has become one of the most popular balance indices in the literature. Given a bifurcating tree $T$, it is defined as the sum, over all internal nodes $v$ in $T$, of the absolute value of the difference between the numbers of descendant leaves of the pair of children of $v$ (even so, there exists a recent extension to multifurcating trees, see \citep{Mir2018}). Its popularity springs from several sources. First of all, its antiquity: it is one of the first balance indices found in the literature, dating back to 1982. Secondly, the way it measures the ``global imbalance'' by adding the ``local imbalances'' of each internal node in $T$ is fairly intuitive. Finally, it has been classified as one of the most powerful tree shape statistics in goodness-of-fit tests of probabilistic models of phylogenetic trees \citep{Agapow02,KiSl:93,Matsen06}, as well as one of the most shape-discriminant balance indices \citep{hayati}.

Due to this popularity, the statistical properties of the Colless index under several probabilistic models have been thoroughly studied \citep{Blum2006a,CMR13,Ford,Heard1992} as well as its maximum \citep{Mir2018} and minimum \citep{MinColless} values. The characterization of this last value, as well as that of the trees attaining it, apart from recent turns out to be rather complex and fails to shed light on the intuitive concept of balance. Indeed, other balance indices, such as the total cophenetic index \citep{MRR} and the rooted quartet index \citep{CMR} classify as ``most balanced'' trees only those that are maximally balanced, in the sense that the imbalance of each internal node is either $0$ or $1$. Even though these trees are effectively considered to be ``most balanced'' by the Colless index, they are seldom the only ones being so considered.

In this manuscript, we introduce a modification of the Colless index that offers some benefits over the original definition, consisting in squaring the difference of the number of descendant leaves to each child of an internal node instead of considering its absolute value. On the one hand, we have been able to compute both its expected value and its variance under the Yule and uniform probabilistic models for phylogenetic trees. In contrast, notice that the expected value of the Colless index under the uniform model is still unknown in the literature. On the other hand, its maximum and minimum values are attained exactly at the caterpillars and the maximally balanced trees, respectively, and the proofs of these results are rather easy ---more so when compared to those concerning the Colless index. Furthermore, it proves to be less prone to have ties between different trees than any other balance index in the literature is, as well as more shape-discriminant than any of the balance indices tested in \citep{hayati} are.

Before leaving the Introduction, we want to note that, even though the Colless index, as well as other indices, was invented for its application to the description and analysis of phylogenetic trees, it is a shape index, i.e. one whose value does not depend on the specific labels associated to the leaves of the tree, but on its underlying topological features. Thus, in the rest of this manuscript we will restrict ourselves to unlabeled trees.

\section{Preliminaries}

\subsection{Trees}

In this paper, by a \emph{tree} $T$ we always mean a \emph{bifurcating rooted tree}, that is, a directed tree with one, and only one, node of in-degree 0 (called the \emph{root} of the tree) and all its nodes of out-degree either 0 (the \emph{leaves}, forming the set $L(T)$) or 2 (the \emph{internal nodes}, forming the set $V_{\mathit{int}}(T)$).
For every $n\geq 1$,  we denote by $\mathcal{T}^*_n$  the set of (isomorphism classes of) trees with $n$ leaves. 

Let $T$ be a tree. If there exists an edge from a node $u$ to a node $v$  in $T$, we say that $v$ is a  \emph{child} of $u$ and that $u$ is the \emph{parent} of $v$. Notice that, since $T$ is bifurcating,  all internal nodes of $T$ have exactly two children. In addition, if there exists a path from a node $u$ to a node $v$  in $T$, we say that $v$ is a \emph{descendant} of $u$. For every node $v$ of $T$, we denote by $\kappa_T(v)$ the number of its  descendant leaves. If $n\geq 2$,  the \emph{maximal pending subtrees} of $T$ are the pair of subtrees rooted at the children of its root. We shall denote the fact that $T_1$ and $T_2$ are the  maximal pending subtrees of $T$ by writing $T=T_1\star T_2$. This notation is commutative, that is 
$T_1\star T_2=T_2\star T_1$.

For every $n\geq 1$, the \emph{comb} with $n$ leaves, $K_n$, is the unique tree in $\TT_n$ all whose internal nodes have different numbers of descendant leaves; cf. Figure \ref{treetop}.(a).

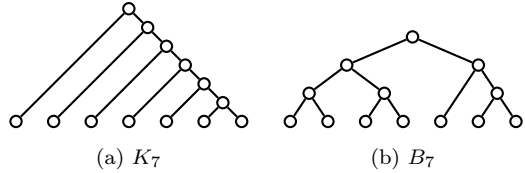
\begin{figure}[htb]
\begin{center}
\begin{tikzpicture}[thick,>=stealth,scale=0.25]
\draw(0,0) node [trepp] (1) {};
\draw(2,0) node [trepp] (2) {};
\draw(4,0) node [trepp] (3) {};
\draw(6,0) node [trepp] (4) {};
\draw(8,0) node [trepp] (5) {}; 
\draw(10,0) node [trepp] (6) {};
\draw(12,0) node [trepp] (7) {};
\draw(11,1) node[trepp] (a) {};
\draw(10,2) node[trepp] (b) {};
\draw(9,3) node[trepp] (c) {};
\draw(8,4) node[trepp] (d) {};
\draw(7,5) node[trepp] (e) {};
\draw(6,6) node[trepp] (r) {};
\draw  (a)--(6);
\draw  (a)--(7);
\draw  (b)--(a);
\draw  (b)--(5);
\draw  (c)--(b);
\draw  (c)--(4);
\draw  (d)--(3);
\draw  (d)--(c);
\draw  (e)--(d);
\draw  (e)--(2);
\draw  (r)--(e);
\draw  (r)--(1);
\draw(6,-2) node {\footnotesize (a) $K_7$};
\end{tikzpicture}
\quad
\begin{tikzpicture}[thick,>=stealth,scale=0.25]
\draw(0,0) node [trepp] (1) {};
\draw(2,0) node [trepp] (2) {};
\draw(4,0) node [trepp] (3) {};
\draw(6,0) node [trepp] (4) {};
\draw(8,0) node [trepp] (5) {}; 
\draw(10,0) node [trepp] (6) {};
\draw(12,0) node [trepp] (7) {};
\draw(1,1.5) node[trepp] (a) {};
\draw(5,1.5) node[trepp] (b) {};
\draw(3,3) node[trepp] (c) {};
\draw(11,1.5) node[trepp] (d) {};
\draw(10,3) node[trepp] (e) {};
\draw(6.5,4.5) node[trepp] (r) {};
\draw  (r)--(c);
\draw  (r)--(e);
\draw  (c)--(a);
\draw  (c)--(b);
\draw  (a)--(1);
\draw  (a)--(2);
\draw  (b)--(3);
\draw  (b)--(4);
\draw  (e)--(d);
\draw  (e)--(5);
\draw  (d)--(6);
\draw  (d)--(7);
\draw(6,-2) node {\footnotesize (b) $B_7$};
\end{tikzpicture}

\end{center}
\caption{\label{treetop} (a) The comb $K_7$ with 7 leaves; (b) The maximally balanced tree $B_7$ with 7 leaves.}
\end{figure}

\subsection{The Colless index and the maximally balanced trees}

Given a tree $T$ and an internal node $v \in V_{\mathit{int}}(T)$ with children $v_1$ and $v_2$, the \emph{balance value} of $v$ is $bal_T(v) = |\kappa_T(v_1)-\kappa_T(v_2)|$. The \emph{Colless index} \cite{Colless82} of a tree $T\in \TT_n$ is the sum of the balance values of its internal nodes:
$$
C(T) = \sum\limits_{v \in V_{\mathit{int}}(T)}  bal_T(v).
$$

An internal node $v$ is \emph{balanced} when $bal_T(v) \leq 1$, i.e. when its two children have $\lceil \kappa_T(v)/2 \rceil$ and $\lfloor \kappa_T(v)/2 \rfloor$ descendant leaves, respectively. 
A tree  is \emph{maximally balanced} if all its internal nodes are balanced (cf. Figure \ref{treetop}.(b)). Recursively, a  bifurcating tree is maximally balanced if its root is balanced and  its two maximal pending subtrees are maximally balanced. This easily implies that, for every $n \in \mathbb{N}$, there exists a unique maximally balanced tree with $n$ leaves, which we denote by $B_n$.  

The maximum Colless index in $\TT_n^*$ is reached exactly at the \emph{comb} $K_n$. The fact that $C(K_n)$ is maximum  was already hinted at by Colless in \cite{Colless82}, but to our knowledge a formal proof that $C(K_n)>C(T)$ for every  $T\in \TT_n^*\setminus\{K_n\}$ was not provided until \cite[Lem. 1] {Mir2018}. As to the  minimum Colless index in $\TT_n^*$, it is proved in \cite[Thm. 1]{MinColless} that it is achieved at the maximally balanced tree $B_n$,
although (unlike the situation with the maximum Colless index) for almost every $n\in \NN_{\geq 1}$ there exist other trees in $\TT^*_n$ with minimum Colless index (see \cite[Cor. 7]{MinColless}). 
If we write $n=\sum_{j=1}^\ell 2^{m_j}$, with $\ell\geq 1$ and $m_1,\ldots,m_\ell\in \NN$ such that $m_1>\cdots>m_\ell$, then 
\begin{equation} 
C(B_n)= \sum_{j=2}^\ell 2^{m_j}(m_1 - m_j - 2(j-2)).
\label{eqn:CBn}
\end{equation} 
For a proof, see Thm. 2 in \cite{MinColless}.

\subsection{Phylogenetic trees}

A  \emph{phylogenetic tree} on a set $X$ is a (rooted and bifurcating) tree with its leaves bijectively labeled by the elements of $X$.  We shall denote by $\TT_X$ the space of (isomorphism classes of) phylogenetic trees on $X$. When the specific set of labels $X$ is irrelevant and only its cardinality $|X|=n$ matters, we shall identify $X$ with the set $\{1,\ldots,n\}$, we shall write $\TT_n$ instead of $\TT_X$, and we shall call the members of this set  \emph{phylogenetic trees with $n$ leaves}.   


A \emph{probabilistic model of phylogenetic trees} $P_{n}$, $n \geq 1$,  is a family of probability mappings $P_n:\TT_n\to [0,1]$, each one sending each phylogenetic tree in $\TT_n$ to its probability under this model. 

The two most popular probabilistic models of phylogenetic trees are the \emph{Yule}, or  \emph{Equal-Rate Markov},  \emph{model} \cite{Harding71,Yule} and the \emph{uniform}, or \emph{Proportional to Distinguishable Arrangements}, \emph{model} \cite{CS,Rosen78}. 
The \emph{Yule model} produces bifurcating phylogenetic trees on $[n]$ through the following stochastic process: starting with a single node,  at every step a leaf is chosen randomly and uniformly and it is replaced by a pair of sister leaves; when the desired number $n$ of leaves is reached, the labels are assigned randomly and uniformly to these leaves. The probability $P_{Y,n}(T)$ of each $T\in \TT_n$ under this model is the probability of being obtained through this process. As to the \emph{uniform model}, it assigns the same probability to all trees $T\in\TT_n$, which is then $P_{U,n}=1/(2n-3)!!$. For more information on these two models, see \cite[\S 3.2]{Steel:16}.

\section{Main theoretical results}

The \emph{Quadratic Colless index}, \emph{Q-Colless index} for short, of a  bifurcating tree $T$ is the sum of the squared balance values of its internal nodes:
$$
\QC(T) = \sum\limits_{v \in V_{\mathit{int}}(T)}  bal_T(v)^2 = \sum\limits_{v \in V_{\mathit{int}}(T)} (\kappa_T(v_1) - \kappa_T(v_2) )^2,
$$
where $v_1$ and $v_2$ denote the children of each $v \in V_{\mathit{int}}(T)$.

For instance, the trees depicted in Figure \ref{treetop} have Q-Colless indices  $\QC(K_7)= 55$ and  $\QC(B_7)=2$. As we shall see, these are the maximum and minimum values of $\QC$ on $\TT_7$.

It is straightforward to check that the Q-Colless index satisfies the following recurrence; cf.  \cite{Rogers1993} for the corresponding recurrence for the ``classical'' Colless index.

\begin{lemma}\label{lem:rec}
For every $T\in \TT^*_n$, with $n\geq 2$, if $T=T_k\star T'_{n-k}$, with $T_k\in \TT^*_k$ and $T'_{n-k}\in \TT^*_{n-k}$, then
$$
\QC(T)=\QC(T_k)+\QC(T'_{n-k})+(n-2k)^2.
$$
\end{lemma}

The Colless index and the Q-Colless index satisfy the following relation.

\begin{lemma}\label{lem:CvsQC}
For every $T\in\BT^*_n$, $\QC(T) \geq C(T)$ and the equality holds if and only if $T$ is maximally balanced.
\end{lemma}

\begin{proof}
By definition,
$$
\QC(T) = \sum_{u\in V_{int}(T)} \bal_T(u)^2 \geq \sum_{u\in V_{int}(T)} \bal_T(u) = C(T)
$$
because $\bal_T(u)\in\NN$ for all $u\in V_{int}(T)$. This inequality is an equality if, and only if, each $\bal_T(u)$ is either 0 or 1, and, by definition, this only happens in the maximally balanced trees.
\end{proof}

\subsection{Extremal values}

In this subsection we prove that, according to the Q-Colless index, the most balanced trees are exactly the maximally balanced trees  and  the most unbalanced trees are exactly the combs.

\begin{theorem}\label{thm:min}
The minimum of the  Q-Colless index on $\BT^*_n$ is always reached at the maximally balanced tree $B_n$, and only at this tree. Moreover, $\QC(B_n)=C(B_n)$ and hence this minimum value is given by Eqn. (\ref{eqn:CBn}).
\end{theorem}

\begin{proof}
Let  $T\in\BT^*_n$. By \cite[Thm. 1]{MinColless}, we know
that $C(T)\geq C(B_n)$. Therefore, by Lemma \ref{lem:CvsQC},
$$
\QC(T) \geq C(T) \geq C(B_n) = \QC(B_n)
$$
and therefore $\QC(B_n)$ is minimum on $\BT^*_n$. Furthermore, the first inequality is strict if $T \neq B_n$, and therefore $\QC(T)> \QC(B_n)$ if $T\neq B_n$.
\end{proof}

\begin{theorem}\label{thm:max}
The maximum of the  Q-Colless index on $\BT^*_n$ is always reached at the comb $K_n$, and only at this tree, and it is equal to.
$$
\QC(K_n) = \binom{n}{3}+\binom{n-1}{3}.
$$
\end{theorem}

\begin{proof}
The formula for $\QC(K_n)$ comes from the fact that the balance values of the internal nodes of $K_n$ are $\{0,1,\ldots,n-2\}$ and therefore
$$
\QC(K_n) = \sum_{i=1}^{n-2} i^2=\frac{(n-1)(n-2)(2n-3)}{6}=\binom{n}{3}+\binom{n-1}{3}.
$$

We prove now the maximality assertion in the statement  by induction on the number $n$ of leaves. For $n \in \{1,2,3\}$, the assertion is obviously true because in these cases $\BT^*_n$ consists of a single tree. Assume now that $n\geq 4$ and that, for every $m<n$, $\QC(K_{m})>\QC(T_m)$ for every $T_m\in \BT^*_m\setminus\{K_m\}$. Let $T\in \BT^*_n$ and let $T_{n_1}$ and $T_{n-n_1}$ be its two maximal pending subtrees, with $T_{n_1}\in \BT^*_{n_1}$ and $T_{n-n_1}\in \BT^*_{n-n_1}$ and, say, $n_1\leq n/2$. In this way, by Lemma \ref{lem:rec},
$$
\QC(T)=\QC(T_{n_1}) + \QC(T_{n-n_1}) + (n - 2n_1)^2.
$$
We want to prove that $\QC(K_n)\geq \QC(T)$ and that the equality  holds only when $T=K_n=K_1\star K_{n-1}$. Since, by induction,
$\QC(K_{n_1})\geq \QC(T_{n_1})$ and $\QC(K_{n-n_1})\geq \QC(T_{n-n_1})$
and the corresponding equalities hold only when $T_{n_1}=K_{n_1}$ and
$T_{n-n_1}=K_{n-n_1}$, it is enough to prove that
$$
\QC(K_n) \geq  \QC(K_{n_1}) + \QC(K_{n-n_1}) + (n-2n_1)^2,
$$
i.e., that
\begin{align*} 
& \frac{(n-1)(n-2)(2n-3)}{6} \\
& \quad \geq \frac{(n_1-1)(n_1-2)(2n_1-3)}{6}+
\frac{(n-n_1-1)(n-n_1-2)(2n-2n_1-3)}{6}\\
& \hspace*{1cm}
+(n-2n_1)^2,
\end{align*}
for every $1\leq n_1\leq n/2$, and that the equality  holds only when $n_1=1$.

Consider now the function $\kappa:[1,n/2]\to\RR$, defined as
\begin{align*} 
\kappa(x)& = \frac{1}{6}\Big((x-1)(x-2)(2x-3) + (n-x-1)(n-x-2)(2n-2x-3)\\
&\hspace*{1cm} + 6(n-2x)^2\Big)\\
& = (n+1)x^2-n(n+1)x+\frac{1}{6}\big(2 n^3  - 3 n^2  + 13 n   - 12\big)
\end{align*}
The graph of this function is a convex parabola with vertex at  $x=n/2$.
Therefore, the maximum value of $\kappa$ on the interval $[1,n/2]$ is reached  at $x=1$,
which is exactly what we wanted to prove.
\end{proof}

By \cite[Cor. 5]{MinColless}, 
$$
\QC(B_n)=C(B_n)<\min\{n/2,2^{\lceil\log_2(n)\rceil}/3\}
$$
and therefore the range of values of $\QC$ on $\TT^*_n$ goes from below this bound to $\binom{n}{3}+\binom{n-1}{3}$ and hence its width grows in $n^3/3$, one order of magnitude larger than the range of the Colless index.

\subsection{Statistics under the uniform and the Yule model}

Let $\QC_n$ be the random variable that chooses a phylogenetic tree $T \in \TT_n$ and computes $\QC(T)$.

\begin{theorem}\label{thm:U}
For every $n\geq 1$:
\begin{enumerate}[(a)]
\item The expected value of $\QC_n$ under the uniform model is
$$
E_U(\QC_n)=\binom{n+1}{2}\cdot \frac{(2n-2)!!}{(2n-3)!!}-n(2n-1).
$$

\item The variance of $\QC_n$ under the uniform model is
\begin{align*}
& \sigma_U^2(\QC_n)  =\frac{2}{15}(2n-1) (7 n^2+9n-1)\binom{n+1}{2}\\
&\qquad
-\frac{1}{8}(5 n^2 + n + 2)\binom{n+1}{2} \frac{(2n-2)!!}{(2n-3)!!} 
-\binom{n+1}{2}^2\left(\frac{(2n-2)!!}{(2n-3)!!}\right)^2.
\end{align*}
\end{enumerate}
\end{theorem}

Regarding the Yule model, we have the following result. In it, $H_n$ and $H_n^{(2)}$ denote, respectively, the $n$-th \emph{harmonic number} and  \emph{second order harmonic number}:
$$
H_n=\sum_{i=1}^n \frac{1}{i},\quad H_n^{(2)}=\sum_{i=1}^n \frac{1}{i^2}.
$$

\begin{theorem}\label{thm:Y}
For every $n\geq 1$:
\begin{enumerate}[(a)]
\item The expected value of $\QC_n$ under the Yule model is
$$
E_Y(\QC_n)=n(n+1)-2nH_n.
$$

\item The variance of $\QC_n$ under the Yule model is
$$
\sigma_Y^2(\QC_n)=\frac{1}{3} n \big(n^3-8n^2+50n-1-30 H_n-12 n H_n^{(2)}\big).
$$
\end{enumerate}
\end{theorem}

We prove these theorems in the Appendix at the end of the paper.\smallskip

Using Stirling's approximation for large factorials it is easy to prove that
$$
\frac{(2n-2)!!}{(2n-3)!!}\sim \sqrt{\pi n}
$$
(see, for instance, \cite[Rem. 2]{vardelta}).
Moreover, it is known (see, for instance,  \cite{Knuth2}) that
$$
H_n\sim \ln(n),\quad  
H_n^{(2)}\sim \frac{\pi^2}{6}.
$$
Then, from the last two theorems we obtain the following limit behaviours:
$$
\begin{array}{ll}
\displaystyle E_U(\QC_n)\sim \frac{\sqrt{\pi}}{2} n^{5/2} &\displaystyle 
 \sigma_U(\QC_n) \sim \sqrt{\frac{14}{15}} n^{5/2}  \\[2ex]
 \displaystyle  E_Y(\QC_n)\sim n^2 &\displaystyle 
 \sigma_Y(\QC_n) \sim \frac{1}{\sqrt{3}} n^{2} \\
 \end{array}
$$
So, both under the Yule and the uniform models, the Q-Colless index satisfies that the expected value and the standard deviation grow with $n$ in the same order. This is in contrast with the Colless index, for which it only happens under the uniform model (see \cite{Blum2006a} for details):
$$
\begin{array}{ll}
\displaystyle E_U(C_n)\sim \sqrt{\pi}  n^{3/2} &\displaystyle 
 \sigma_U(C_n) \sim \sqrt{\frac{10-3\pi}{3}} n^{3/2}  \\[2ex]
 \displaystyle  E_Y(C_n)\sim n\log(n) &\displaystyle 
 \sigma_Y(C_n) \sim \sqrt{\frac{18-6\log(2)-\pi^2}{6}} n .\\
 \end{array}
$$

\section{Numerical results}

Since the range of values of the Q-Colless index on $\TT_n^*$ is wider than those of the Colless index $C$, the Sackin index $S$ or the total cophenetic index $\Phi$ (see Table \ref{tab:range}), our intuition told us that the probability of two trees with the same number of leaves having the same Q-Colless index would be smaller than for these other balance indices. To simplify the language, when a balance index $I$ takes the same value on two trees in the same space $\mathcal{T}^*_n$, we call it a \emph{tie}.
Of course, since for $n\geq 12$ the range of possible $\QC$ values is narrower than the number of trees in $\TT_n^*$ (see \cite[Table 3.3]{fel:04} for the cardinality of $\TT_n^*$ for small values of $n$), the pigeonhole principle implies that the Q-Colless index cannot avoid ties for large numbers of leaves.

\begin{table}
    \centering
    \begin{tabular}{r|ll}
    \mbox{Index} & \mbox{Minimum} & \mbox{Maximum}\\ \hline
    $C$     &  $\Theta(n)$ & $\binom{n-1}{2}$\\[1ex]
    $S$     &  $\Theta(n\log(n))$ & $\binom{n+1}{2}-1$\\[1ex]
    $\Phi$  &  $\Theta(n^2)$  & $\binom{n}{3}$\\[1ex]
    $\QC$ &  $\Theta(n)$   &  $\binom{n}{3}+\binom{n-1}{3}$
\end{tabular}
    \caption{Range of values of the Colless index $C$, the Sackin index $S$, the total cophenetic index $\Phi$ and the Q-Colless index $\QC$}
    \label{tab:range}
\end{table}

To check the discriminative power of $\QC$ with respect to $C$, $S$, and $\Phi$, we have computed the probability of tie $p_n(I)$ for these four balance indices $I$ and for number of leaves of $n$ between $n=4$ and $n=20$.

More concretely, first of all, for every balance index $I=C,S,\Phi,\QC$ and for $n=4,\ldots,20$, we have considered all pairs of different trees $(T_1,T_2)$ in $\mathcal{T}^*_n\times \mathcal{T}^*_n$ and we have calculated the number $n_I$ of such pairs of trees such that $I(T_1)=I(T_2)$. Finally, we have computed the probability $p_n(I)$ as $p_n(I)=\frac{n_I}{\binom{|\mathcal{T}^*_n|}{2}}$, where $|\mathcal{T}^*_n|$ is the cardinal of the set $\mathcal{T}^*_n$. 
The results obtained are shown in figure~\ref{SMALL.TIES}. The Q-Colless balance index is the balance index with the least probability of a tie.

In relation with this last point, another way to assess the discriminating skill of an index is to evaluate its power to distinguish between dissimilar trees, and compare it with that of other shape indices. In their paper \cite{hayati}, the authors (whom we thank for their support with the software provided in the article) develop a new resolution function to evaluate the power of tree shape statistics when it comes to discriminate between dissimilar trees (based on the Laplacian matrix of the tree, which allows for less spatial and time complexity in the operations), and then test it together with the usual resolution function based on the NNI metric. Therefore, they are able to rank some balance indices according to their power in discriminating all possible phylogenetic trees on the same number of leaves.

We have performed the same experiment on the same data (which was provided along with \cite{hayati}). It turns out that the $\QC$ performs better than all the other tested indices do, including the Saless index \cite{hayati}, a linear combination of the Sackin and Colless indices which was introduced in the same article and performed best when tested under the NNI metric --- although not with the resolution function proposed in the article, under which it was the Colless index that performed better. We present here the two tables, the first of them computing the score under the NNI distance (bigger values represent more power), and the second one under their proposed resolution function (lower values represent more power).

\begin{table}
\begin{center}
\begin{tabular}{c|cccccccc}
Number of leaves & Colless & Sackin	& Variance  & $I_2$     & $B_1$      & $B_2$    & Saless & Q-Colless  \\ \hline
5                & 1       & 1      & 1         & 1      & 1      & 1      & 1& 1          \\
6                & 0.8157  & 0.8510 & 0.8144    & 0.7611 & 0.7546 & 0.8705 & 0.8315 & 0.8709 \\
7                & 0.9251  & 0.9303 & 0.9023    & 0.8844 & 0.8649 & 0.9254 & 0.9297 & 0.9360 \\
8                & 0.9255  & 0.9122 & 0.8753    & 0.8612 & 0.8326 & 0.9113 & 0.9235 & 0.9218 \\
9                & 0.9184  & 0.9208 & 0.8826    & 0.8539 & 0.8324 & 0.907  & 0.9224 & 0.9302 \\
10               & 0.941   & 0.9380 & 0.8985    & 0.8545 & 0.8326 & 0.9085 & 0.9426 & 0.9475 \\
11               & 0.9531  & 0.9514 & 0.9102    & 0.8552 & 0.8375 & 0.9132 & 0.9551& 0.9604 \\
12               & 0.9533  & 0.9523 & 0.9086    & 0.8504 & 0.8311 & 0.9045 & 0.9556 & 0.9632 \\
13               & 0.9541  & 0.9542 & 0.9078    & 0.8416 & 0.8247 & 0.8992 & 0.9567 & 0.9657 \\
14               & 0.9552  & 0.9548 & 0.9070    & 0.8374 & 0.82   & 0.8902 & 0.9575 & 0.967  \\
15               & 0.9546  & 0.9544 & 0.9049    & 0.8298 & 0.813  & 0.8826 & 0.9569 & 0.9674 \\
16               & 0.9543  & 0.9541 & 0.9034    & 0.8265 & 0.8089 & 0.8743 & 0.9564 & 0.9677 \\
17               & 0.9534  & 0.9534 & 0.9006    & 0.8199 & 0.8024 & 0.8678 & 0.9555 & 0.9679 
\end{tabular}
\caption{Scaled resolution scores for shape indices on the NNI distance matrix. The value of the resolution is between $0$ and $1$. Higher values represent more discriminating power.}
\end{center}
\end{table}

\begin{table}
\begin{center}
\begin{tabular}{c|cccccccc}
Number of leaves & Colless& Sackin & Variance& $I_2$     & $B_1$     & $B_2$     & Q-Colless \\ \hline
7                & 0.0984 & 0.0937 & 0.1082  & 0.1115 & 0.1178 & 0.0989 & 0.0948 \\
8                & 0.0808 & 0.0955 & 0.111   & 0.0893 & 0.1164 & 0.0965 & 0.0941 \\
9                & 0.0507 & 0.0566 & 0.0662  & 0.068  & 0.0797 & 0.0653 & 0.0558 \\
10               & 0.0327 & 0.0379 & 0.0471  & 0.0535 & 0.0629 & 0.0451 & 0.0357 \\
11               & 0.0222 & 0.0255 & 0.0326  & 0.0458 & 0.0511 & 0.0348 & 0.0236 \\
12               & 0.0183 & 0.0217 & 0.0282  & 0.0429 & 0.0473 & 0.0304 & 0.0194 \\
13               & 0.016  & 0.0185 & 0.0238  & 0.0413 & 0.0441 & 0.0283 & 0.0163 \\
14               & 0.0147 & 0.0170 & 0.0217  & 0.04   & 0.0421 & 0.0265 & 0.0147 \\
15               & 0.0137 & 0.0157 & 0.0197  & 0.039  & 0.0404 & 0.0256 & 0.0134 \\
16               & 0.013  & 0.0148 & 0.0184  & 0.038  & 0.0389 & 0.0247 & 0.0126 \\
17               & 0.0123 & 0.014  & 0.017   & 0.037  & 0.0375 & 0.0238 & 0.0118 \\
18               & 0.0117 & 0.0132 & 0.016   & 0.0358 & 0.0361 & 0.0229 & 0.0111 \\
19               & 0.0112 & 0.0127 & 0.015   & 0.0347 & 0.0349 & 0.0222 & 0.0105 \\
20               & 0.0107 & 0.012  & 0.0141  & 0.0339 & 0.0338 & 0.0217 & 0.01 \\
21               & 0.0102 & 0.0114 & 0.0133  & 0.0329 & 0.0327 & 0.0209 & 0.01
\end{tabular}
\caption{Scaled resolution scores for shape indices on the resolution function presented in \cite{hayati}. The value of the resolution is between $0$ and $1$. Lower values represent more discriminating power.}
\end{center}
\end{table}

\begin{figure}
    \centering
    \includegraphics[width=10cm]{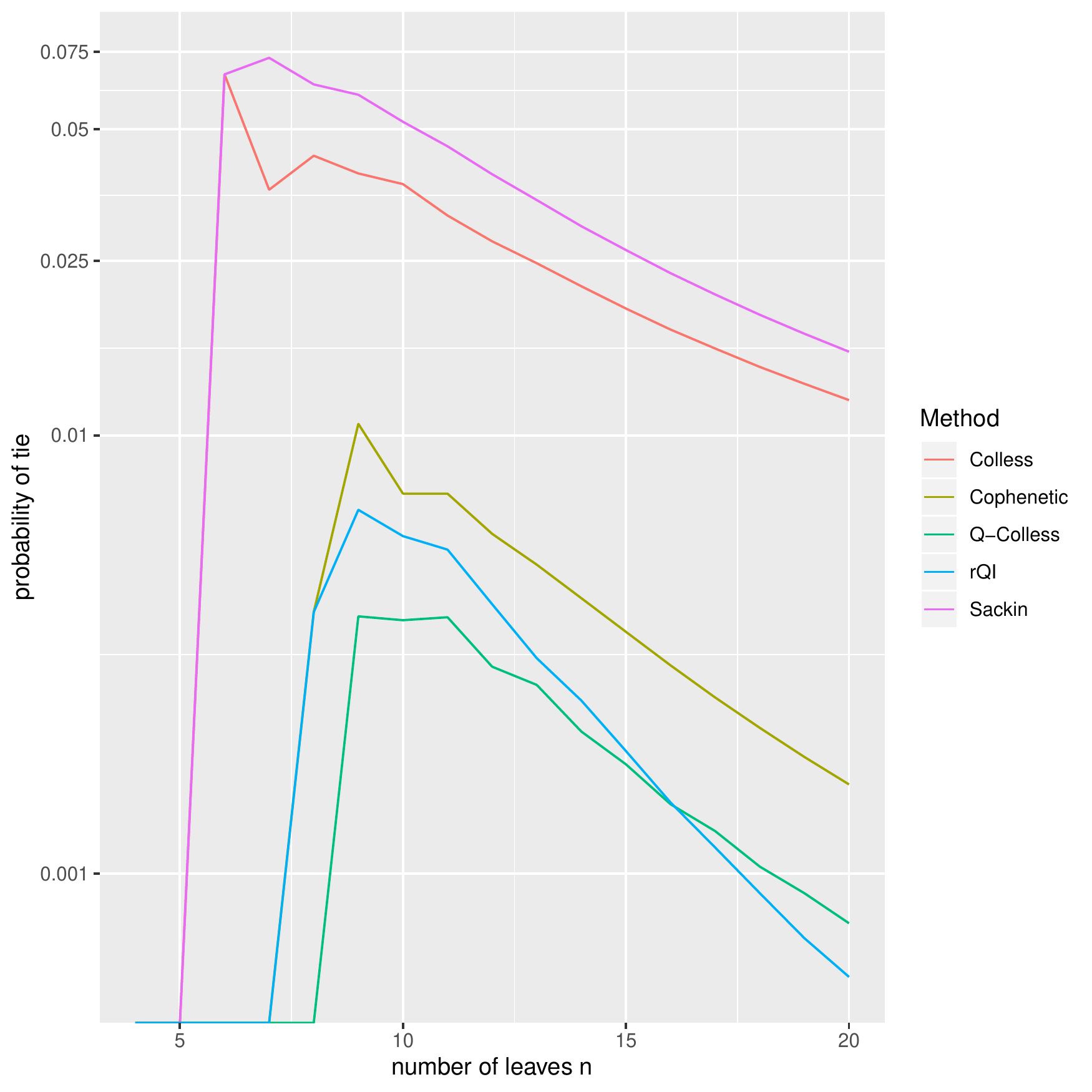}
    \caption{Probability of tie using the Colless, Cophenetic, Quadratic Colless, rQI and Sackin balance indices as function of the trees'number of leaves $n$, for $n=4,\ldots, 20$. }
    \label{SMALL.TIES}
\end{figure}

\section{Conclusions}

The Colless index \citep{Colless82} is one of the oldest and most popular balance indices appearing in the literature. Its number of cites more than doubles that of the second most cited balance index in Google Scholar, the Sackin index. Nevertheless, it presents some drawbacks related to the difficult characterisation of the trees that achieve its minimum value ---which clashes with the intuition that only the maximally balanced trees should be considered the most balanced bifurcating trees--- and the fact that its moments under one of the most widely used probabilistic models for bifurcating phylogenetic trees, the uniform model, are still unknown.

In this paper we have presented an alternative to the Colless index that captures both its intuitive definition and its statistical benefits. In the first part of this manuscript we have proved that its extremal values are attained exactly by the trees that are usually considered to be the ``most'' and ``least'' balanced family of bifurcating trees, respectively. This contrasts vividly with the Colless and Sackin indices, whose minimum value, although being always reached by the maximally balanced trees, is seldom attained only by it; although the Colless index was defined in 1982 \citep{Colless82}, these characterizations have been only very recently found \citep{MinColless, Fischer2018}. We have thus shown that the range of values of the Quadratic Colless index, $O(n^3)$, is bigger than that of the original Colless index, $O(n^2)$, on pair with that of the total cophenetic index.

Then, we have proceeded to the computation of both the expected value and the variance under the Yule and the Uniform models of the Q-Colless index. We want to remark to the reader that the expected value and the variance of the Colless index in its original definition are, under the uniform model, still unknown. So, in this regard the Quadratic Colless index presents an improvement over the original measure of balance.

Finally, we have empirically shown that it possesses more discriminatory power than the original Colless index does by, firstly, computing the probability of producing a tie between a pair of trees for numbers of leaves up to $20$ and, then, testing it under the metrics provided in \cite{hayati}. In both cases, it has systematically been one of the best performing measures, being often superior to the Colless and Sackin indices.

\section*{Acknowledgements}  This work was partially supported by the Spanish Ministry of Economy and Competitiveness and the European Regional Development Fund  through project  PGC2018-096956-B-C43 (MINECO/FEDER).

\section*{Appendices}

\subsection*{A.1 Proof of Theorem \ref{thm:Y}}

The following lemma summarizes Lemma 16 in \cite{MRR} and  Lemma 2 in \cite{CMR13}.

\begin{lemma}\label{lem:IY}
Let $I:\bigcup_{n\geq 1} \TT_n\to \mathbb{R}$ be a mapping satisfying the following two conditions:
\begin{itemize}
\item  It is  invariant under phylogenetic tree isomorphisms  and relabelings of leaves.

\item  There exists a symmetric mapping $f_I:\NN_{n\geq 1}\times \NN_{n\geq 1}\to \RR$ such that, for every pair of phylogenetic trees $T,T'$ on disjoint sets of taxa $X,X'$, respectively,
$$
I(T\star T')=I(T)+I(T')+f_I(|X|,|X'|).
$$
\end{itemize}
For every $n\geq 1$, let $I_n$ and $I_n^2$ be the random variables that choose a tree $T \in \TT_n$ and compute $I(T)$ and $I(T)^2$, respectively.  Then, for every $n\geq 2$, their expected values under the Yule model are:
\begin{align*}
& E_Y(I_n) =
\frac{1}{n-1}\sum_{k=1}^{n-1}\big(2E_Y(I_k)+f_I(k,n-k)\big)\\[1ex]
& E_Y(I_n^2) =\frac{1}{n-1}\sum_{k=1}^{n-1} \Big(2E_Y(I^2_k)  + 4f_I(k,n-k)E_Y(I_k)+2E_Y(I_k)E_Y(I_{n-k})\\
&\qquad\qquad \qquad\quad +f_I(k,n-k)^2\Big).
\end{align*}
\end{lemma}

\begin{claim}\label{thm:EY}
For every $n\geq 1$, the expected value of $\QC_n$ under the Yule model is
$$
E_Y(\QC_n)=n(n+1)-2nH_n.
$$
\end{claim}

\begin{proof}
By Lemma \ref{lem:IY}.(a),
\begin{align*}
E_Y(\QC_n) & =\frac{2}{n-1}\sum_{k=1}^{n-1} E_Y(\QC_k)+\frac{1}{n-1}\sum_{k=1}^{n-1} (n-2k)^2\\
& = \frac{2}{n-1}\sum_{k=1}^{n-1} E_Y(\QC_k)+\frac{1}{3}n(n-2)\\
& = \frac{2}{n-1}E_Y(\QC_{n-1})+\frac{n-2}{n-1}\Big(\frac{2}{n-2}\sum_{k=1}^{n-2} E_Y(\QC_k)\Big)+\frac{1}{3}n(n-2)\\
& = \frac{2}{n-1}E_Y(\QC_{n-1})+\frac{n-2}{n-1}\Big(E_Y(\QC_{n-1})-\frac{1}{3}(n-1)(n-3)\Big)\\
&\qquad\qquad+\frac{1}{3}n(n-2)\\
& = \frac{n}{n-1}E_Y(\QC_{n-1})+n-2
\end{align*}
Dividing this equation by $n$ and setting $X_n=E_Y(\QC_n)/n$, we obtain the equation
$$
X_n=X_{n-1}+1-\frac{2}{n} 
$$
whose solution with initial condition $X_1=E_Y(\QC_1)=0$ is
$$
X_n=\sum_{k=2}^n\Big(1-\frac{2}{k} \Big)=n+1-2H_n
$$
and hence, finally,
$$
E_Y(\QC_n)=nX_n=n(n+1)-2nH_n.
$$
\end{proof}

\begin{claim}
For every $n\geq 1$, the variance of $\QC_n$ under the Yule model is
$$
\sigma_Y^2(\QC_n)=\frac{1}{3} n \big(n^3-8n^2+50n-1-30 H_n-12 n H_n^{(2)}\big).
$$
\end{claim}

\begin{proof}
We shall compute the variance $\sigma_Y^2(\QC_n)$
by means of the identity
\begin{equation}
\sigma_Y^2(\QC_n)=E_Y(\QC_n^2)-E_Y(\QC_n)^2
\label{eqn:varY}
\end{equation}
where the value of $E_Y(\QC_n)$ is given by Theorem \ref{thm:EY}. What remains is to compute  $E_Y(\QC_n^2)$. Now, by Lemma \ref{lem:IY}.(b),
\begin{align*}
 E_Y(\QC_n^2)  &=
\frac{1}{n-1}\sum_{k=1}^{n-1}\Big(
2E_Y(\QC_k^2) +   (n-2k)^4 \\
&\qquad  
+4  (n-2k)^2 E_Y(\QC_k) +2  E_Y(\QC_k)  E_Y(\QC_{n-k})
\Big)\\
& =\frac{2}{n-1}\sum_{k=1}^{n-1}E_Y(\QC_k^2) + \frac{1}{n-1}\sum_{k=1}^{n-1}(n-2k)^4 \\ 
&\qquad  
+ \frac{4}{n-1}\sum_{k=1}^{n-1}(n-2k)^2k (k+1-2H_k) 
 \\ 
&\qquad  
 +\frac{2}{n-1}\sum_{k=1}^{n-1}k(n-k)(k+1-2H_k)(n-k+1-2H_{n-k})
\end{align*}
Let us denote by $T_n$ the independent term in this equation, so that this equation can be written as 
\begin{align*}
E_Y(\QC_n^2) & =\frac{2}{n-1}\sum_{k=1}^{n-1}E_Y(\QC_k^2)+T_n\\
& =\frac{2}{n-1}E_Y(\QC_{n-1}^2)+\frac{n-2}{n-1}\cdot \frac{2}{n-2}\sum_{k=1}^{n-2}E_Y(\QC_k^2)+T_n\\
& =\frac{2}{n-1}E_Y(\QC_{n-1}^2)+\frac{n-2}{n-1}(E_Y(\QC_{n-1}^2)-T_{n-1})+T_n\\
& =\frac{n}{n-1}E_Y(\QC_{n-1}^2)+T_n-\frac{n-2}{n-1}T_{n-1}
\end{align*}

Dividing this equation by $n$ and setting $Y_n=E_Y(\QC_n^2)/n$, we obtain the equation
\begin{equation}
Y_n=Y_{n-1}+\frac{1}{n}\Big(T_n-\frac{n-2}{n-1}T_{n-1}\Big).
\label{eqn:EYC22-mareT}
\end{equation}

We want to compute now the independent term in this equation as an explicit expression in $n$. To do that, we first compute  the three sums that form $T_n$. On the one hand,
\begin{equation} 
\frac{1}{n-1}\sum_{k=1}^{n-1}(n-2k)^4=\frac{1}{15}n(n-2)(3 n^2- 6 n-4).
\label{eqn:EYC22-TI1}
\end{equation} 
On the other hand,
\begin{align}
& \frac{4}{n-1}\sum_{k=1}^{n-1}(n-2k)^2k (k+1-2H_k)\nonumber\\
&\qquad =\frac{4}{n-1}\Bigg(\sum_{k=1}^{n-1}(n-2k)^2k (k+1)-2(n-2)^2\sum_{k=1}^{n-1}kH_k\nonumber\\
&\qquad\qquad\qquad +16(n-3)\sum_{k=1}^{n-1}\binom{k}{2}H_k-48\sum_{k=1}^{n-1}\binom{k}{3}H_k\Bigg)\nonumber
\\
&\qquad = \frac{4}{n-1}\Bigg(\frac{1}{15}(n - 1)n(n + 1)(2 n^2 - 5 n + 2)
-2(n-2)^2\binom{n}{2}\Big(H_n-\frac{1}{2}\Big)\nonumber\\
&\qquad\qquad\qquad+16(n-3)\binom{n}{3}\Big(H_n-\frac{1}{3}\Big)-48\binom{n}{4}\Big(H_n-\frac{1}{4}\Big)\Bigg)\nonumber\\
&\qquad = \frac{2}{45}n(n - 2)(12 n^2 + 16 n + 9)-\frac{4}{3} n^2(n - 2)H_n
\label{eqn:EYC22-TI2}
\end{align}
using, in the second last equality above, that 
\begin{equation}
\sum_{k=1}^{n-1} \binom{k}{m}H_k=\binom{n}{m+1}\Big(H_n-\frac{1}{m+1}\Big);
\label{eqn:Knuth}
\end{equation}
see Eqn. (6.70) in \cite{Knuth2}.

As to the third sum,
\begin{align}
& \frac{2}{n-1}\sum_{k=1}^{n-1} k(n-k)(k+1-2H_k)(n-k+1-2H_{n-k})\nonumber\\
&\qquad=\frac{2}{n-1}\Bigg[\sum_{k=1}^{n-1}k(k+1)(n-k)(n-k+1)\nonumber\\
&\qquad\qquad-2\sum_{k=1}^{n-1}k(n-k)(n-k+1)H_k -2\sum_{k=1}^{n-1}k(n-k)(k+1)H_{n-k}\nonumber\\
&\qquad\qquad+4\sum_{k=1}^{n-1}k(n-k)H_kH_{n-k}\Bigg]\nonumber\\
&\qquad=\frac{2}{n-1}\Bigg[\sum_{k=1}^{n-1}k(k+1)(n-k)(n-k+1)-4\sum_{k=1}^{n-1}k(n-k)(n-k+1)H_k\nonumber\\
&\qquad\qquad+4n\sum_{k=1}^{n-1}kH_kH_{n-k}-4\sum_{k=1}^{n-1}k^2H_kH_{n-k}\Bigg]\nonumber\\
&\qquad=\frac{2}{n-1}\Bigg[\sum_{k=1}^{n-1}k(k+1)(n-k)(n-k+1)\nonumber\\
&\qquad\qquad-4\sum_{k=1}^{n-1}\Big(6\binom{k}{3}-4(n-1)\binom{k}{2}+n(n-1)k\Big)H_k\nonumber\\
&\qquad\qquad+4n\sum_{k=1}^{n-1}kH_kH_{n-k}-4\sum_{k=1}^{n-1}k^2H_kH_{n-k}\Bigg]\nonumber\\
&\qquad=\frac{2}{n-1}\Bigg[4\binom{n+3}{5} -24\binom{n}{4}\Big(H_n-\frac{1}{4}\Big)\nonumber\\
&\qquad\qquad+16(n-1)\binom{n}{3}\Big(H_n-\frac{1}{3}\Big)
-4n(n-1)\binom{n}{2}\Big(H_n-\frac{1}{2}\Big)\nonumber\\
&\qquad\qquad+4n\binom{n+1}{2}\big(H_{n+1}^2-H_{n+1}^{(2)}-2H_{n+1}+2\big)\nonumber\\
&\qquad\qquad-\frac{4}{3}\binom{n+1}{2}\Big((2n+1)(H_{n+1}^2-H_{n+1}^{(2)}) -\frac{13n+5}{3}H_{n+1}+\frac{71n+37}{18}\Big)\Bigg]\nonumber\\
&\qquad=\frac{1}{270}n(18 n^3 + 303 n^2 + 1163 n + 98)-\frac{2}{9}n(n+1)(3n+16)H_n\nonumber\\
&\qquad\qquad 
+\frac{4}{3}n(n+1)(H_{n+1}^2-H_{n+1}^{(2)})
\label{eqn:EYC22-TI3}
\end{align}
using, in the second last equality above,  Eqn. (\ref{eqn:Knuth}) and the identities
\begin{align*}
 \sum_{k=1}^{n-1} kH_kH_{n-k} &=\binom{n+1}{2}\big(H_{n+1}^2-H_{n+1}^{(2)}-2H_{n+1}+2\big)\\
 \sum_{k=1}^{n-1} k^2H_kH_{n-k} & =\frac{n(n+1)}{6}\Big[(2n+1)(H_{n+1}^2-H_{n+1}^{(2)})\\
&\quad -\frac{13n+5}{3}H_{n+1}+\frac{71n+37}{18}\Big]
\end{align*}
proved in  \cite{WGW}.

So, 
\begin{align*}
& T_n= \frac{1}{15}n(n-2)(3 n^2- 6 n-4)\\
&\qquad\qquad  +\frac{2}{45}n(n - 2)(12 n^2 + 16 n + 9)-\frac{4}{3} n^2(n - 2)H_n\\
&\qquad\qquad  +\frac{1}{270}n(18 n^3 + 303 n^2 + 1163 n + 98)-\frac{2}{9}n(n+1)(3n+16)H_n\\
&\qquad\qquad 
+\frac{4}{3}n(n+1)(H_{n+1}^2-H_{n+1}^{(2)})\\
&\qquad = \frac{1}{270}n(216 n^3 - 9 n^2 + 1031 n + 26)
-\frac{2}{9}n (9 n^2 + 7 n + 16)H_n\\
&\qquad\qquad 
+\frac{4}{3}n(n+1)(H_{n+1}^2-H_{n+1}^{(2)})
\end{align*}
and, hence, the independent term in Eqn. (\ref{eqn:EYC22-mareT})
is
\begin{align*}
 & \frac{1}{n}\Big(T_n-\frac{n-2}{n-1}T_{n-1}\Big)\\
 &\quad =\frac{1}{n}\Bigg[\frac{1}{270}n(216 n^3 - 9 n^2 + 1031 n + 26)
-\frac{2}{9}n (9 n^2 + 7 n + 16)H_n\\
&\qquad\qquad 
+\frac{4}{3}n(n+1)(H_{n+1}^2-H_{n+1}^{(2)})\\
&\qquad\qquad 
-\frac{n-2}{n-1}\Bigg(\frac{1}{270}(n-1)(216 (n-1)^3 - 9 (n-1)^2 + 1031 (n-1) + 26)\\
&\qquad\qquad 
-\frac{2}{9}(n-1) (9 (n-1)^2 + 7 (n-1) + 16)H_{n-1}\\
&\qquad\qquad 
+\frac{4}{3}(n-1)n(H_{n}^2-H_{n}^{(2)})\Bigg)\Bigg]\\
&\quad =\frac{1}{n}\Bigg[\frac{1}{270}n(216 n^3 - 9 n^2 + 1031 n + 26)\\
&\qquad\qquad 
-\frac{2}{9}n (9 n^2 + 7 n + 16)H_{n-1}-\frac{2}{9}(9 n^2 + 7 n + 16)\\
&\qquad\qquad 
+\frac{4}{3}n(n+1)(H_{n}^2-H_{n}^{(2)})+\frac{8}{3}nH_{n-1}+\frac{8}{3}\\
&\qquad\qquad 
-\frac{1}{270}(n-2)(216 n^3 - 657 n^2 + 1697 n - 1230)\\
&\qquad\qquad 
+\frac{2}{9}(n-2) (9 n^2 - 11 n + 18)H_{n-1}\\
&\qquad\qquad 
-\frac{4}{3}(n-2)n(H_{n}^2-H_{n}^{(2)})\Bigg]\\
&\quad=\frac{1}{n}\Big(\frac{1}{3}(12 n^3 - 28 n^2 + 47 n - 30)-8 (n^2 - n + 1)H_{n-1} \\
&\qquad\qquad
+4n(H_{n}^2-H_{n}^{(2)})\Big)
\\
&\quad=4 n^2 - \frac{28}{3} n + \frac{47}{3} - \frac{10}{n}
-8(n-1)H_{n-1}-\frac{8H_{n-1}}{n}+
4H_{n}^2-4H_{n}^{(2)} 
\end{align*}

 The solution of Eqn. (\ref{eqn:EYC22-mareT}) with initial condition $Y_1=E_Y(QC_1^2)=0$ is
 \begin{align*}
Y_n & =\sum_{k=2}^n \frac{1}{k}\Big(T_k-\frac{k-2}{k-1}T_{k-1}\Big)\\
& = \sum_{k=2}^n\Big(4 k^2 - \frac{28}{3} k + \frac{47}{3} - \frac{10}{k}
-8(k-1)H_{k-1}-\frac{8H_{k-1}}{k}+
4H_{k}^2-4H_{k}^{(2)}\Big)\\
& = \sum_{k=1}^{n-1}\Bigg(4(k+1)^2 - \frac{28}{3}(k+1) + \frac{47}{3} - \frac{10}{k+1}
-8kH_{k}-\frac{8H_{k}}{k+1}+
4H_{k+1}^2-4H_{k+1}^{(2)}\Bigg)\\
& \stackrel{(*)}{=} \frac{1}{3}(4 n^3 - 8 n^2 + 35 n - 31)-10(H_n-1)\\
&\qquad
-8\binom{n}{2}\Big(H_n-\frac{1}{2}\Big)-
4(H_n^2-H_n^{(2)})\\
&\qquad+4\big((n+1)H_{n}^2-(2n+1)H_{n}+2n-1\big)\\
&\qquad
-4\big((n+1)H_{n}^{(2)}-H_{n}-1\big)\\
& = \frac{1}{3}(4 n^3 - 2 n^2 + 53 n - 1)-2 (2 n^2 + 2 n + 5)H_n
+4 n(H_{n}^2-H_{n}^{(2)})
\end{align*}
where, in the second last equality (marked with ($*$)) we have used Eqn. (\ref{eqn:Knuth}) and
the identities
$$
\sum_{k=1}^{n-1} \frac{H_k}{k+1}=\frac{1}{2}(H_n^2-H_n^{(2)})
$$
(cf. Eqn. (6.71) in \cite{Knuth2}) and
\begin{align*}
& \sum_{k=1}^{n-1} H_k^2=nH_{n}^2-(2n+1)H_{n}+2n\\
& \sum_{k=1}^{n-1} H_k^{(2)}=nH_n^{(2)}-H_n
\end{align*}
(see \cite[\S 1.2.7]{Knuth1}).

Therefore, finally
\begin{align*}
& E_Y(\QC_n^2)=nY_n\\
&\quad =\frac{n}{3}(4 n^3 - 2 n^2 + 53 n - 1)-2n (2 n^2 + 2 n + 5)H_n
+4n^2(H_{n}^2-H_{n}^{(2)})
\end{align*}
and 
\begin{align*}
 &\sigma_Y^2(\QC_n)=E_Y(\QC_n^2)-E_Y(\QC_n)^2\\
 &\quad =\frac{1}{3} n \big(n^3-8n^2+50n-1-30 H_n-12 n
   H_n^{(2)}\big)
\end{align*}
as we claimed.
\end{proof}

\subsection*{A.2 Proof of Theorem \ref{thm:U}}

To simplify the notations, for every $n\geq 2$ and for every $1\leq k\leq n-1$, set
$$
C_{k,n-k}\coloneqq \frac{1}{2}\binom{n}{k}\frac{(2k-3)!! (2(n-k)-3)!!}{(2n-3)!!}.
$$

The proof of the following lemma is identical to the proof of Lemma \ref{lem:IY} given in the references provided in the previous subsection, simply replacing the probabilities under the Yule model by  probabilities under the uniform model. We leave the details to the reader.

\begin{lemma}\label{lem:I2}
Let $I:\bigcup_{n\geq 1} \TT_n\to \mathbb{R}$ be a mapping satisfying the same conditions as in the statement of Lemma \ref{lem:IY} and, for every $n\geq 1$, let $I_n$ and $I_n^2$ be the random variables that choose a tree $T \in \TT_n$ and compute $I(T)$ and $I(T)^2$, respectively.  Then, for every $n\geq 2$, their expected values under the uniform model are:
\begin{align*}
& E_U(I_n) =
\sum_{k=1}^{n-1}C_{k,n-k}\big(2E_U(I_k)+f_I(k,n-k)\big)\\[1ex]
& E_U(I_n^2)  =
\sum_{k=1}^{n-1}C_{k,n-k} \left( \vphantom{\sum}
 2E_U(I^2_{k}) +   f_I(k,n-k)^2  
\right. \\ &\quad  \displaystyle  \qquad\qquad\qquad\qquad  \left. 
+4  f_I(k,n-k) E_U(I_{k}) +2  E_U(I_{k})  E_U(I_{n-k}).
\vphantom{\sum}\right) 
\end{align*}
\end{lemma}

In the proofs provided in this subsection we shall use the following technical lemmas. They are proved in the Section SN-4 of the Supplementary Material of \cite{vardelta}; Lemma \ref{prop:gralsol} is Proposition 6 in that paper.

\begin{lemma}\label{lem:previsums1}
For every $n\geq 2$:
\begin{enumerate}[(a)]
\item $\displaystyle \sum_{k=1}^{n-1} C_{k,n-k}=1$
\item For every $m\geq 1$,
$$
\sum_{k=1}^{n-1} C_{k,n-k} \binom{k}{m}=\frac{1}{2}\binom{n}{m}\Big(1-\frac{m-1}{n-1}\cdot\frac{(2m-3)!!}{(2m-2)!!}\cdot \frac{(2n-2)!!}{(2n-3)!!}\Big).
$$
\end{enumerate}
\end{lemma}

\begin{lemma}\label{lem:previsums2}
For every $n\geq 2$,
\begin{enumerate}[(a)]
\item $\displaystyle \sum_{k=1}^{n-1} C_{k,n-k}\cdot \frac{(2k-2)!!}{(2k-3)!!}=\frac{1}{2}\cdot \frac{(2n-2)!!}{(2n-3)!!}+\frac{1}{4}\big(2H_{2n-2}- H_{n-1}-2\big)$.

\item For every $m\geq 1$,
$$
\sum_{k=1}^{n-1} C_{k,n-k} \binom{k}{m} \frac{(2k-2)!!}{(2k-3)!!}=\frac{1}{2}\binom{n}{m}\Big(\frac{(2n-2)!!}{(2n-3)!!}-\frac{(2m-2)!!}{(2m-3)!!}\Big).
$$
\end{enumerate}
\end{lemma}

\begin{lemma}\label{prop:gralsol}
The solution $X_n$ of the equation
$$
X_n=2\sum_{k=1}^{n-1} C_{k,n-k} X_k+\sum_{l=1}^r a_l\binom{n}{l}+\frac{(2n-2)!!}{(2n-3)!!}\sum_{l=1}^s b_l\binom{n}{l}
$$
with given initial condition $X_1$ is
$$
X_n=\sum_{l=1}^{s+1} \widehat{a}_l\binom{n}{l}+\frac{(2n-2)!!}{(2n-3)!!}\sum_{l=1}^r \widehat{b}_l\binom{n}{l}
$$
with
\begin{align*}
& \widehat{a}_1= X_1-a_1\\
& \widehat{a}_{l}   =\frac{l\cdot (2l-2)!!}{(2l-3)!!}\Big(\frac{b_{l}}{l} +\frac{b_{l-1}}{l-1}\Big),\quad l=2,\ldots,s\\
&\widehat{a}_{s+1}   =\frac{(s+1)\cdot (2s)!!}{s\cdot (2s-1)!!}\cdot b_s\\
&\widehat{b}_l=\frac{(2l-3)!!}{(2l-2)!!}\cdot a_l,\quad l=1,\ldots,r
\end{align*}
\end{lemma}

\begin{claim}\label{thm:EU}
For every $n\geq 1$, the expected value  of $\QC_n$ under the uniform model is
$$
E_U(\QC_n)=\binom{n+1}{2}\cdot \frac{(2n-2)!!}{(2n-3)!!}-n(2n-1).
$$
\end{claim}

\begin{proof}
By Lemma \ref{lem:I2}.(a),
\begin{align*} 
& E_U(\QC_n) =2\sum_{k=1}^{n-1}C_{k,n-k} E_U(\QC_k)+\sum_{k=1}^{n-1}C_{k,n-k}(n-2k)^2\\
&\qquad= 2\sum_{k=1}^{n-1}C_{k,n-k} E_U(\QC_k)+n^2\sum_{k=1}^{n-1}C_{k,n-k}-4(n-1)\sum_{k=1}^{n-1}C_{k,n-k}k\\
&\qquad\qquad +8\sum_{k=1}^{n-1}C_{k,n-k}\binom{k}{2}\\
&\qquad= 2\sum_{k=1}^{n-1}C_{k,n-k} E_U(\QC_k)+n^2-2n(n-1)\\
&\qquad\qquad +4\binom{n}{2}\Big(1-\frac{1}{2(n-1)}\cdot \frac{(2n-2)!!}{(2n-3)!!}\Big)\\
&\qquad= 2\sum_{k=1}^{n-1}C_{k,n-k} E_U(\QC_k)+2\binom{n}{2}+n-n\cdot \frac{(2n-2)!!}{(2n-3)!!}
\end{align*}
where in the second last equality we have used Lemma \ref{lem:previsums1}. Therefore, by Lemma \ref{prop:gralsol} and using that $E_U(\QC_1)=0$, we have that
\begin{align*}
E_U(\QC_n) & =\Big(\binom{n}{2}+n\Big)\frac{(2n-2)!!}{(2n-3)!!}-\Big(4\binom{n}{2}+n\Big)\\
& = \binom{n+1}{2}\cdot \frac{(2n-2)!!}{(2n-3)!!}-n(2n-1)
\end{align*}
as we claimed.
\end{proof}

\begin{claim}
For every $n\geq 1$, the variance of $\QC_n$ under the uniform model is
\begin{align*}
\sigma_U^2(\QC_n) & =\frac{2}{15}(2n-1) (7 n^2+9n-1)\binom{n+1}{2}\\
&\quad
-\frac{1}{8}(5 n^2 + n + 2)\binom{n+1}{2} \frac{(2n-2)!!}{(2n-3)!!} 
-\binom{n+1}{2}^2\left(\frac{(2n-2)!!}{(2n-3)!!}\right)^2
\end{align*}
\end{claim}

\begin{proof}
To simplify the notations, we shall denote $(2n-2)!!/(2n-3)!!$ by $\AAA_n$. We shall compute the variance $\sigma_U^2(\QC_n)$
by means of the identity
\begin{equation}
\sigma_U^2(\QC_n)=E_U(\QC_n^2)-E_U(\QC_n)^2
\label{eqn:varU}
\end{equation}
where the value of $E_U(\QC_n)$ is given by Theorem \ref{thm:EU}. Now, we must compute  $E_U(\QC_n^2)$. 
By Lemma \ref{lem:I2}.(b),
\begin{align*} 
& E_U(\QC_n^2)  =
\sum_{k=1}^{n-1}C_{k,n-k} \left( \vphantom{\sum} 2E_U(\QC_k^2) +   (n-2k)^4  \right. \\ 
&\qquad\qquad\qquad\qquad  \left. 
+4 (n-2k)^2 E_U(\QC_k) +2  E_U(\QC_k)  E_U(\QC_{n-k})
\vphantom{\sum}\right) \\
& = 2\sum_{k=1}^{n-1}C_{k,n-k}E_U(\QC_k^2)\\
&\qquad +\sum_{k=1}^{n-1}C_{k,n-k}\left[(n-2k)^4+4(n-2k)^2\Big(\binom{k+1}{2}\AAA_k-k(2k-1)\Big)\right.\\
&\qquad\qquad\qquad\qquad+ 2\Big(\binom{k+1}{2}\AAA_k-k(2k-1)\Big)\\
& \qquad\qquad\qquad\qquad\qquad\left.\cdot \Big(\binom{n-k+1}{2}\AAA_{n-k}-(n-k)(2(n-k)-1)\Big)\right]
\\
& = 2\sum_{k=1}^{n-1}C_{k,n-k}E_U(\QC_k^2)
\\
&\qquad +\sum_{k=1}^{n-1}C_{k,n-k}\Big((n-2k)^4-4(n-2k)^2k(2k-1)\\
& \qquad\qquad\qquad\qquad\qquad +2k(2k-1)(n-k)(2(n-k)-1)\Big)\\
&\qquad +\sum_{k=1}^{n-1}C_{k,n-k}\left[4(n-2k)^2\binom{k+1}{2}\AAA_k-2\binom{n-k+1}{2}k(2k-1)\AAA_{n-k}\right.\\
& \qquad\qquad\qquad\qquad\qquad\left.-2\binom{k+1}{2}(n-k)(2(n-k)-1)\AAA_k\right]\\
&\qquad +2\sum_{k=1}^{n-1}C_{k,n-k}\binom{k+1}{2}\binom{n-k+1}{2}\AAA_{k}\AAA_{n-k}\\
& = 2\sum_{k=1}^{n-1}C_{k,n-k}E_U(\QC_k^2)
\\
&\qquad -\sum_{k=1}^{n-1}C_{k,n-k}\Big(8 k^4+16 (n - 1) k^3 - 2 (12n^2-6n - 1) k^2 -  (2 n - 8 n^3)k - n^4\Big)\\
&\qquad +\sum_{k=1}^{n-1}C_{k,n-k}\left[4(n-2k)^2\binom{k+1}{2}-4\binom{k+1}{2}(n-k)(2(n-k)-1)\right]\AAA_k\\
&\qquad +2\sum_{k=1}^{n-1}C_{k,n-k}\binom{k+1}{2}\binom{n-k+1}{2}\AAA_{k}\AAA_{n-k}
\end{align*}
\begin{align} 
& = 2\sum_{k=1}^{n-1}C_{k,n-k}E_U(\QC_k^2)
\nonumber\\
&\qquad -\sum_{k=1}^{n-1}C_{k,n-k}\Bigg[192\binom{k}{4}+96(n+2)\binom{k}{3}-4(12 n^2 - 30 n - 5)\binom{k}{2}\nonumber\\
&\qquad\qquad +(8 n^3 - 24 n^2 + 26 n - 6)k  - n^4 \Bigg]\nonumber\\
&\qquad +\sum_{k=1}^{n-1}C_{k,n-k}\Bigg[96\binom{k}{4}+156\binom{k}{3}-4 (n^2-n-16)\binom{k}{2}\nonumber\\
&\qquad\qquad -4(n^2-n-1) k\Bigg]\AAA_k\nonumber\\
&\qquad +2\sum_{k=1}^{n-1}C_{k,n-k}\binom{k+1}{2}\binom{n-k+1}{2}\AAA_{k}\AAA_{n-k}.\label{eqn:masterEUQC2}
\end{align}

Let us compute the independent term in this equation.
The first sum can be computed using Lemma~\ref{lem:previsums1}:
\begin{align*}
& \sum_{k=1}^{n-1}C_{k,n-k}\Big(192\binom{k}{4}+96(n+2)\binom{k}{3}-4(12 n^2 - 30 n - 5)\binom{k}{2}\\ & \qquad\qquad +(8 n^3 - 24 n^2 + 26 n - 6)k  - n^4 \Big)\\
& \quad =96\binom{n}{4}\Big(1-\frac{3}{n-1}\cdot \frac{5!!}{6!!}\cdot \AAA_n\Big)\\
& \qquad\qquad +48(n+2)\binom{n}{3}\Big(1-\frac{2}{n-1}\cdot \frac{3!!}{4!!}\cdot \AAA_n\Big)\\
& \qquad\qquad -2(12 n^2 - 30 n - 5)\binom{n}{2}\Big(1-\frac{1}{2(n-1)}\cdot\AAA_n\Big)\\
& \qquad\qquad +\frac{1}{2}(8 n^3 - 24 n^2 + 26 n - 6)n-n^4\\
&\quad=(3 n-2) n^3-\frac{n(15n^2-15n+4)}{4}\cdot \AAA_n.
\end{align*}

The second sum in this independent term can be computed using Lemma~\ref{lem:previsums2}:
\begin{align*}
& \sum_{k=1}^{n-1}C_{k,n-k}\left[96\binom{k}{4}+156\binom{k}{3}-4 (n^2-n-16)\binom{k}{2}-4(n^2-n-1) k\right]\AAA_k\\
&\qquad = 48\binom{n}{4}\Big(\AAA_n-\frac{6!!}{5!!}\Big)
+78\binom{n}{3}\Big(\AAA_n-\frac{4!!}{3!!}\Big)\\
&\qquad\qquad -2 (n^2-n-16)\binom{n}{2}(\AAA_n-2)-2(n^2-n-1)n(\AAA_n-1)\\
& \qquad=n^3(n+1)\AAA_n-\frac{2n(33n^3-13n^2-12n+7)}{15}.
\end{align*}

Finally, the third sum in the independent term of this equation can be computed as follows:
\begin{align*}
& 2\sum_{k=1}^{n-1}C_{k,n-k}\binom{k+1}{2}\binom{n-k+1}{2}\frac{(2k-2)!!}{(2k-3)!!}\frac{(2n-2k-2)!!}{(2n-2k-3)!!}\\
&\quad=\sum_{k=1}^{n-1}
\frac{n!(2k-3)!! (2(n-k)-3)!!k(k+1)(n-k)(n-k+1)2^{k-1}(k-1)!2^{n-k-1}(n-k-1)!}{k!(n-k)!(2n-3)!!2^2(2k-3)!!(2(n-k)-3)!!}\\
&\quad=\frac{n!2^{n-4}}{(2n-3)!!}\sum_{k=1}^{n-1}(k+1)(n-k+1)\\
& \quad =\frac{n!2^{n-3}(n-1)(n+1)(n+6)}{(2n-3)!!6}=
\frac{n+6}{8}\cdot \binom{n+1}{3}\cdot\AAA_n.
\end{align*}

So, the independent term of Eqn. (\ref{eqn:masterEUQC2}) is
\begin{align*}
& \frac{n(15n^2-15n+4)}{4}\cdot \AAA_n -(3 n-2) n^3 \\
&\qquad\qquad +n^3(n+1)\AAA_n-\frac{2n(33n^3-13n^2-12n+7)}{15} +
\frac{n+6}{8}\cdot \binom{n+1}{3}\cdot\AAA_n\\
& = \frac{n(49 n^3 + 234 n^2 - 181 n +42)}{48}\cdot\AAA_n
-\frac{n(111 n^3 - 56 n^2 - 24 n + 14)}{15}\\
& = \left(3 n+36 \binom{n}{2}+66 \binom{n}{3}+\frac{49}{2}\binom{n}{4}\right)\alpha_n -3 n-78\binom{n}{2}-244\binom{n}{3}-\frac{888}{5}\binom{n}{4}
\end{align*}
and, hence, Eqn. (\ref{eqn:masterEUQC2}) simplifies to
\begin{align*}
& E_U(\QC_n^2)  = 2\sum_{k=1}^{n-1}C_{k,n-k}E_U(\QC_k^2)-3n-78\binom{n}{2}-244\binom{n}{3}-\frac{888}{5}\binom{n}{4}\\ &\qquad\qquad +\left(3 n+36 \binom{n}{2}+66 \binom{n}{3}+\frac{49}{2}\binom{n}{4}\right)\alpha_n.
\end{align*}
This equation can be solved using Lemma \ref{prop:gralsol} and the fact that $E_U(\QC_1^2)=0$. Its solution is
\begin{align*}
 & E_U(\QC_n^2)  =  3 n+84\binom{n}{2}+320\binom{n}{3}+360\binom{n}{4}+112\binom{n}{5}\\ &\qquad \qquad
-\left(3 n+39\binom{n}{2}+\frac{183}{2}\binom{n}{3}+\frac{111}{2}\binom{n}{4}\right)\alpha_n \\ &\qquad =  \frac{n}{15}(14 n^4 + 85 n^3 - 60 n^2 + 5 n +1)\\ &\qquad \qquad- \frac{n}{16} (37 n^3 + 22 n^2 -13 n +2)\AAA_n.
\end{align*}

Finally,
\begin{align*}
&\sigma_U^2(\QC_n)=E_U(\QC_n^2)-E_U(\QC_n)^2\\
&\quad =\frac{2}{15}(2n-1) (7 n^2+9n-1)\binom{n+1}{2}\\
&\qquad
-\frac{1}{8}(5 n^2 + n + 2)\binom{n+1}{2} \frac{(2n-2)!!}{(2n-3)!!} 
-\binom{n+1}{2}^2\left(\frac{(2n-2)!!}{(2n-3)!!}\right)^2
\end{align*}
as we claimed.
\end{proof}

\end{document}